\documentclass[11pt,a4paper]{amsart}

\usepackage[T1]{fontenc}
\usepackage[utf8]{inputenc}
\usepackage{amsmath, amssymb, amsfonts, amsthm}
\usepackage[margin=3.0cm]{geometry}
\usepackage[]{setspace}
\usepackage[usenames,dvipsnames]{color}
\usepackage[unicode]{hyperref}
\usepackage{comment}
\usepackage[all,matrix,arrow]{xy}
\usepackage[shortlabels]{enumitem}
\usepackage[table]{xcolor}
\usepackage{tikz}

\usepackage{footnote,todonotes}
\makesavenoteenv{tabular}

\usepackage{array}

\theoremstyle{plain}
\newtheorem{theorem}{Theorem}[section]
\newtheorem{proposition}[theorem]{Proposition}
\newtheorem{lemma}[theorem]{Lemma}
\newtheorem{corollary}[theorem]{Corollary}

\theoremstyle{definition}
\newtheorem{example}[theorem]{Example}
\newtheorem{definition}[theorem]{Definition}
\newtheorem{question}[theorem]{Question}

\theoremstyle{remark}
\newtheorem{remark}[theorem]{Remark}

\newcommand{\bA}{\mathfrak A}
\newcommand{\bB}{\mathfrak B}
\newcommand{\bS}{\mathfrak S}
\newcommand{\bQ}{\mathfrak Q}
\newcommand{\bR}{\mathfrak R}

\DeclareMathOperator{\CSP}{CSP}

\DeclareMathOperator{\GL}{GL}

\numberwithin{equation}{section}

\title[Polynomial-time Tractable Problems over the $p$-adic Numbers]{Polynomial-time Tractable Problems\\ over the $p$-adic Numbers}

\thanks{\today}

\author{Arno Fehm \and Manuel Bodirsky}
\address{Institut f\"{u}r Algebra, Technische Universit\"{a}t Dresden, 01062 Dresden, Germany}
\email{manuel.bodirsky@tu-dresden.de}
\email{arno.fehm@tu-dresden.de}

\begin{document}
 
\begin{abstract} 
We study the computational complexity of fundamental problems over the $p$-adic numbers ${\mathbb Q}_p$ and the $p$-adic integers ${\mathbb Z}_p$.
Gu\'epin, Haase, and Worrell~\cite{GHW19}
proved that checking satisfiability of systems of linear equations 
combined with valuation constraints of the form $v_p(x) = c$ for $p \geq 5$ is NP-complete (both over ${\mathbb Z}_p$ and over ${\mathbb Q}_p$), and left the cases $p=2$ and $p=3$ open. We solve their problem by showing that the problem is NP-complete 
for ${\mathbb Z}_3$ and for ${\mathbb Q}_3$, but that it is in P 
 for ${\mathbb Z}_2$ and for ${\mathbb Q}_2$. 
We also present different polynomial-time algorithms for solvability of systems of linear equations in ${\mathbb Q}_p$ with either constraints of the form $v_p(x) \leq c$ or of the form $v_p(x)\geq c$ for $c \in {\mathbb Z}$.
Finally, we show how our algorithms can be used to decide in polynomial time the satisfiability of systems of (strict and non-strict) linear inequalities 
over ${\mathbb Q}$ together with valuation constraints $v_p(x) \geq c$ for several different prime numbers $p$ simultaneously.
\end{abstract}

\maketitle

\section{Introduction}

\noindent
The satisfiability problem for systems of polynomial equations is an immensely useful computational problem; however, is has a quite bad worst-time complexity: it is NP-hard in arbitrary fields, undecidable over ${\mathbb Z}$~\cite{MatiyasevichDoklady}, not known to be decidable over ${\mathbb Q}$, 
and not known to be in NP for ${\mathbb R}$~\cite{Schaefer2015}.
In contrast, the satisfiability problem for systems of \emph{linear} equations
has a much better computational complexity: 
it can be solved in polynomial time 
over ${\mathbb R}$ and, equivalently, over ${\mathbb Q}$, and even over
${\mathbb Z}$ (see, e.g.,~\cite{Schrijver}). 
It is therefore natural to search for meaningful extensions of the satisfiability problem for linear systems that retain some of the pleasant computational properties; in particular, extensions that remain in the complexity class P. It is also interesting to search for meaningful restrictions of the satisfiability problem for systems of polynomial equations that are no longer computationally hard. 

One of the well-studied expansions of linear systems is the expansion by linear \emph{inequalities}. 
Note that $x \leq y$ can be expressed over 
${\mathbb R}$ by
$\exists z (x+z^2 = y)$ (and it can also be expressed over ${\mathbb Q}$ and ${\mathbb Z}$, but we then need a different formula), so this expansion can also be viewed as a restriction of the mentioned problem for systems of polynomial equations. The satisfiability problem for linear inequalities is known to be NP-complete over ${\mathbb Z}$, but remains in P over ${\mathbb Q}$ and ${\mathbb R}$  (e.g., via the ellipsoid method; see, e.g.,~\cite{Schrijver}). 

Other interesting, but less well-known expansions of the linear existential theory of ${\mathbb Z}$ and ${\mathbb Q}$
come from $p$-adic valuations $v_p$, for $p$ a prime number: 
For $x \in {\mathbb Z}$, one defines
$v_p(x):=\sup\{j:p^j|x\}\in\mathbb{N}\cup\{\infty\}$,
and one extends this to $\mathbb{Q}$ by 
$v_p(\frac{a}{b}) := v_p(a) - v_p(b)$.
The complexity of the satisfiability problem 
for systems of linear equalities combined with 
valuation constraints of the form $v_p(x) = c$
for $c\in\mathbb{Z}$
has been studied 
by Gu\'epin, Haase, and Worrell~\cite{GHW19}. 
Their results show that the problem  
over ${\mathbb Q}$ is in NP, even if the constants $c$ are represented in binary and $p$ is part of the input. 
This is remarkable, because for 
any $x = \frac{a}{b} \in {\mathbb Q}$ that satisfies
$v_p(x) = c >0$, the number $a$ has exponential size in $c$, i.e., doubly exponential size in the input size. So we cannot simply guess and verify a solution in binary representation. 

The results of Gu\'epin, Haase, and Worrell are actually stated in a different setting: they phrase their result over the \emph{$p$-adic numbers}. The $p$-adic valuation gives rise to a (non-archimedean) \emph{absolute value}, defined for $x \in {\mathbb Q}$ 
by   
$|x|_p := p^{-v_p(x)}$.  
The \emph{field 
of $p$-adic numbers} ${\mathbb Q}_p$ 
is the completion of ${\mathbb Q}$ 
with respect to $|\cdot|_p$, similarly as ${\mathbb R}$ is defined to be the completion of ${\mathbb Q}$ with respect to the standard absolute value. The \emph{{ring of} $p$-adic integers} 
is the subring ${\mathbb Z}_p$ of ${\mathbb Q}_p$ 
with domain $\{x \in {\mathbb Q}_p \mid v_p(x) \geq 0\}$,
where $v_p$ denotes the natural extension of the $p$-adic valuation to $\mathbb{Q}_p$. 
Gu\'epin, Haase, and Worrell~\cite{GHW19} 
phrase their mentioned results as satisfiability problems over ${\mathbb Q}_p$; however, the problems are equivalent to the respective problems over ${\mathbb Q}$; see Proposition~\ref{prop:qvsqp}. 
They then use their algorithm to prove that the entire existential theory of ${\mathbb Q}_p$ in a suitable (linear) language is in NP.

Gu\'epin, Haase, and Worrel moreover obtain some hardness results: 
they prove that the satisfiability problem for systems of linear equations over ${\mathbb Q}_p$ and over ${\mathbb Z}_p$ with valuation constraints of the form 
$v_p(x) = c$ is NP-hard for $p \geq 5$. 
They also state:
\emph{``While we believe it to be the case, it remains an open problem whether an NP lower bound can also be established for the cases $p=2,3$.''}~\cite[Remark 23]{GHW19}.  

We solve this problem and prove that satisfiability is NP-complete in the case $p=3$ for both ${\mathbb Q}_p$ and ${\mathbb Z}_p$. For $p = 2$, however, we prove containment in P. Interestingly, our algorithm can also cope with constraints of the form $v_p(x) \geq c$, even if $p$ is larger than $2$ {(Theorem~\ref{thm:alg-geq})}. We also find an algorithm that can test the satisfiability of linear systems for $\mathbb{Q}_p$ in the presence of constraints of the form $v_p(x) \leq c$ {(Proposition~\ref{prop:alg-leq})}; it is the combination of both upper and lower valuation bounds that makes the problem hard.

{Our algorithm can also be used for the satisfiability problem for valuation constraints in combination with linear inequalities over ${\mathbb Q}$.} 
We prove that the satisfiability of systems of (weak and strict) linear inequalities together with various valuation constraints,
for instance of the form $v_p(x) \geq c$, 
can be decided in polynomial time (Theorem \ref{thm:inequalities}). We do allow valuation constraints 
for different primes in the input; we 
allow binary representations of all coefficients in the input. 
The proof uses the 
fact that linear programming is in P~\cite[{Section 13}]{Schrijver}, 
and the 
approximation theorem for finitely many inequivalent absolute values for ${\mathbb Q}$ (\cite[Ch.~XII, Thm.~1.2]{Lang}).

\medskip 
{\bf{Related Work.}}
The computational complexity for satisfiability problems of semilinear expansions 
of linear inequalities over ${\mathbb Q}$ (equivalently: over ${\mathbb R})$ has been studied in~\cite{Essentially-convex}. The results there state that every expansion of the satisfiability problem for linear inequalities by other semilinear relations is NP-hard, unless all relations $R \subseteq {\mathbb Q}^n$ are \emph{essentially convex}, i.e., 
have the property that for any two $a,b \in R$, all but finitely many rational points on the line segment between $a$ and $b$ are also contained in $R$; moreover, if all relations are essentially convex, then the satisfiability problem is in P~\cite[Theorem 5.2]{Essentially-convex}. 
This result has later been generalised to expansions of linear equalities instead of inequalities~\cite{JonssonThapper15}. 
Valuation constraints are clearly not essentially convex; however, they are also not semilinear, and not even semialgebraic, and hence are not covered by the results from~\cite{Essentially-convex} and from~\cite{JonssonThapper15}.

Different computational tasks for the $p$-adic numbers have been studied by Dolzmann and Sturm~\cite{DolzmannSturmPadic}, and more recently by Haase and Mansutti~\cite{HaaseMansutti}: they showed that whether a given system of linear equations with valuation constraints (where the valuation constraints in~\cite{HaaseMansutti} are more expressive than the ones from~\cite{DolzmannSturmPadic}, which are more expressive than ours) 
has a solution in ${\mathbb Q}_p$ for \emph{all} prime numbers $p$ is in coNExpTime.

Another recent results is a polynomial-time algorithm for the \emph{dyadic feasibility problem}~\cite{DyadicLP}, which is the problem of testing the satisfiability of systems of linear inequalities over 
${\mathbb Z}[\frac{1}{2}]$; it is unclear how to reduce this problem to the problems studied here and vice versa.

\section{Preliminaries}

\noindent
We recall some well-known facts about $p$-adic numbers, see e.g.\ \cite{p-adic-book},
and how we treat them from a logic and a computational point of view.
We write ${\mathbb P} \subseteq {\mathbb N}$ for the set of all prime numbers and we let $p\in\mathbb{P}$.

\subsection{$\mathbb{Q}_p$ and $\mathbb{Z}_p$}
As $\mathbb{Q}_p$ is by definition the completion of $\mathbb{Q}$ with respect to the $p$-adic absolute value $|.|_p$, it is a metric space whose topology is the {\em $p$-adic topology}.
The $p$-adic absolute value on $\mathbb{Q}_p$ gives rise to the $p$-adic valuation $v_p(x)=-\log_p|x|_p$. It satisfies the following basic properties:

\begin{lemma}\label{lem:vpq}
For all $a,b \in {\mathbb Q}_p$ we have 
\begin{itemize}
\item $v_p(a \cdot b) = v_p(a) + v_p(b)$, and 
\item $v_p(a+b) \geq \min(v_p(a),v_p(b))$, with equality if $v_p(a) \neq v_p(b)$.
\end{itemize}
\end{lemma}

The set $\mathbb{Z}_p=\{x\in\mathbb{Q}_p:v_p(x)\geq 0\}$ forms a subring of $\mathbb{Q}_p$ called the {\em ring of $p$-adic integers}. Its unique maximal ideal is generated by $p$, and $\mathbb{Z}_p/p^n\mathbb{Z}_p\cong\mathbb{Z}/p^n\mathbb{Z}$ for every $n\in\mathbb{N}$. This implies the following fact, which we will use several times:

\begin{lemma}\label{lem:ac}
For every $x\in\mathbb{Q}_p \setminus \{0\}$
with $n=v_p(x)$ there exists a unique $i\in\{1,\dots,p-1\}$ such that $v_p(x-ip^n)>n$.
\end{lemma}

This further implies that every $p$-adic number has a unique {\em $p$-adic expansion}:

\begin{lemma}\label{lem:adic-exp}
Every $0\neq x\in\mathbb{Q}_p$ with $n=v_p(x)$ is the limit (in the $p$-adic topology) of a unique series of the form
$\sum_{i=n}^\infty x_ip^i$ with $x_i\in\{0,\dots,p-1\}$ for every $i$.
\end{lemma}

As usual, we let 
$$
\mathbb{Z}_{(p)} :=\mathbb{Z}_p\cap\mathbb{Q}=\left\{x\in\mathbb{Q}:v_p(x)\geq0\right\}=\left\{\frac{a}{b}:a,b\in\mathbb{Z},p\nmid b\right\},
$$ 
see Figure~\ref{fig:incl}.

\begin{figure}
\begin{center} 
\begin{tikzpicture}[scale=1.5]
  \node (Qp) at (0,2) {$\mathbb{Q}_p$};
  \node (Zp1) at (-1,1) {$\mathbb{Z}_p$};
  \node (Q) at (1,1) {$\mathbb{Q}$};
  \node (Zp2) at (0,0) {$\mathbb{Z}_{(p)}$};
  \node (Z) at (0,-1) {$\mathbb{Z}$};
\node (R) at (2,2)
{$\mathbb{R}$};

  \draw (Zp1) -- (Qp);
  \draw (Q) -- (Qp);
  \draw (Zp1) -- (Zp2);
  \draw (Q) -- (Zp2);
  \draw (Zp2) -- (Z);
  \draw (Q) -- (R);
\end{tikzpicture}
\end{center} 
\caption{Inclusions between the number domains studied in this article.}
\label{fig:incl}
\end{figure}
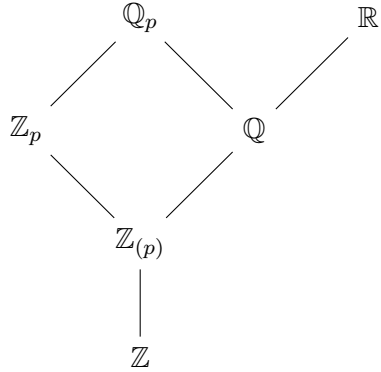

\subsection{The structure ${\mathfrak Q}_p$} 
It will be convenient for some of our results and proofs to take a logic perspective on the $p$-adic numbers; for an introduction to first-order logic, see~\cite{Hodges}. A \emph{signature} is a set $\tau$ of relation and function symbols, each equipped with an \emph{arity}, which is a natural number. 
A \emph{(first-order) structure} $\bS$ of signature $\tau$ consists of a set (the domain, typically denoted by the corresponding capital roman letter $S$), a function $f^{\bS} \colon S^k \to S$ for each function symbol $f \in \tau$ of arity $k \in {\mathbb N}$ (the case $k=0$ is allowed; in this case, we refer to $f$ as a \emph{constant symbol}), and a relation $R^{\bS} \subseteq S^k$ for each relation symbol $R \in \tau$ of arity $k$; we then say that $f$ \emph{denotes} $f^{\bS}$, and $R$ \emph{denotes} $R^{\bS}$.  

A \emph{reduct} of $\bS$ is a structure obtained from $\bS$ by taking a subset of the signature. If $\bR$ is a reduct of $\bS$, then $\bS$ is called an \emph{expansion} of $\bR$. 
A \emph{substructure} of $\bS$ is a structure 
$\bS'$ with the same signature $\tau$ as $\bS$ and domain $S' \subseteq S$ 
such that
 for every function symbol $f \in \tau$ of arity $k$,
 the function 
$f^{\bS'}$ is the restriction of $f^{\bS}$ to $(S')^k$, and for every relation symbol $R \in \tau$ of arity $k$, the relation $R^{\bS'}$ equals $R^{\bS} \cap (S')^k$.

A \emph{first-order $\tau$-formula} is a formula built from first-order quantifiers $\forall, \exists$, Boolean connectives $\wedge, \vee, \neg$, and atomic formulas that are built from variables, the equality symbol $=$, and the symbols from $\tau$ in the usual way; for a proper definition, we refer to any standard introduction to mathematical logic or model theory, such as~\cite{Hodges}.

\begin{remark}
    Often when $p$-adic numbers are treated from a logic perspective, they are introduced as `two-sorted structures', with one sort for the $p$-adic numbers and one sort for the values, i.e., ${\mathbb Z} \cup \{\infty\}$, and a function symbol $v$ for the valuation.
    For our purposes, however, usual first-order structures (as introduced above) are sufficient. 
\end{remark}

We work with the structure ${\mathfrak Q}_p$
which has the domain ${\mathbb Q}_p$ and the signature 
$$\{+,1\} \cup \{\leq^p_c,\geq^p_c ,=^p_c,\neq^p_c,\mid c \in {\mathbb Z}\},$$
where
\begin{itemize}
\item $+$ is a binary function symbol that denotes the addition operation of $p$-adic numbers as introduced above;  
\item $1$ is a constant symbol which denotes $1 \in {\mathbb Z}_{(p)} = {\mathbb Q}_p \cap {\mathbb Z}$ as introduced above;
\item $\leq^p_{c}$ is a unary  relation symbol that denotes the unary relation $\{ x \in {\mathbb Q}_p \mid v_p(x) \leq c\}$; 
$\geq^p_c$, $=^p_{c}$, and $\neq^p_c$ are defined analogously. 
\end{itemize}
Sometimes, we specify structures as tuples; e.g., we write 
$$ 
 {\mathfrak Q}_p = ({\mathbb Q}_p;+,1,(\leq^p_c)_{c \in {\mathbb Z}},(\geq^p_c)_{c \in {\mathbb Z}},(=^p_c)_{c \in {\mathbb Z}},(\neq^p_c)_{c \in {\mathbb Z}})
$$
and do not distinguish between function and relation symbols and the respective functions and relations. 
Atomic formulas
that are built from the relations $\leq_c^p$, $\geq_c^p$, $=_c^p$, and $\neq_c^p$ will be called \emph{valuation constraints}. 
For $c \in {\mathbb Z}$, we also use the symbols $<^p_{c}$ as a shortcut for $\leq_{c-1}^p$, and $>^p_{c}$ as a shortcut for $\geq_{c+1}^p$.


\subsection{Primitive Positive Formulas and CSPs}
A formula is called \emph{primitive positive} if it is of the form 
$$\exists x_1,\dots,x_n (\psi_1 \wedge \cdots \wedge \psi_m)$$
where $\psi_1,\dots,\psi_m$ are atomic.
In \emph{primitive existential} formulas, 
$\psi_1,\dots,\psi_m$ are allowed to be negated atomic formulas as well, and \emph{existential} formulas are disjunctions of primitive existential formulas.
We use the concepts of 
primitive positive (and primitive existential, and existential) sentences, theories, definitions, definability, etc, as in the case of first-order logic (see, e.g.,~\cite{Hodges}), but restricting to primitive positive (primitive existential, and existential) formulas.

The computational problem of deciding the truth of a given primitive positive sentence $\varphi$ in a fixed structure $\bS$ is called the \emph{constraint satisfaction problem (CSP)} of $\bS$. 
We refer to the quantifier-free part $\psi_1 \wedge \dots \wedge \psi_m$ of $\varphi$ as the 
\emph{instance} of $\CSP(\bS)$ (i.e., the existential quantifiers will be left implicit), 
and a satisfying assignment to the variables will also be called a \emph{solution} to $\varphi$.

If the signature of $\bS$ is infinite, then the computational problem is not yet well-defined, because we still have to specify how to represent the symbols from the signature in the input;  the choice of the representation can have an impact on the complexity of the CSP. For the structure ${\bQ}_p$ introduced above, 
a natural representation is to represent the relation symbols $\leq_c^p$, $\geq^p_c$, 
$=_c^p$, and $\neq^p_c$ 
by the binary encoding of $p\in\mathbb{P}$ and $c \in {\mathbb Z}$.
Note that $v_p(x) \leq c$ 
holds if and only if
$v_p(x) \leq v_p(p^c)$. 
It will turn out that in all of our polynomial-time tractability results, 
it suffices to store $c$ in binary (which makes $p^c$ a doubly exponentially large number). 
The hardness results, however, always make use of only finitely many symbols in the signature, and hence hold independently from the choice of the representation. We will therefore allow binary representations for the values $c$ in the valuation constraints, since this allows the strongest formulations of our results. 

We will determine the computational complexity of $\CSP(\bS)$ for all
reducts of ${\mathfrak Q}_p$ (Theorem~\ref{thm:class} and~\ref{thm:class2}).

\subsection{Primitive positive interpretations}
Primitive positive interpretations can be used to obtain complexity reductions between CSPs. For $d \geq 1$, a $d$-dimensional \emph{primitive positive interpretation} of a structure
$\bA$ in a structure $\bB$ is given by a partial function $I$ from $B^d$ to $A$ 
such that the preimages under $I$ of the following sets are primitively  positively definable in $\bB$: 
\begin{itemize}
\item $A$ and the equality relation $=_A$ on $A$, 
\item each relation of $\bA$, and 
\item each graph of a function of $\bA$.
\end{itemize}

\begin{lemma}[see, e.g.,~{\cite[{Theorem 3.1.4}]{Book}}]\label{lem:pp-int}
    Let $\bA$ 
    be a structure with a finite signature and a primitive positive interpretation in a structure $\bB$. Then $\bB$ has a reduct $\bB'$
    with a finite signature such that there is a polynomial-time reduction from $\CSP(\bA)$ to $\CSP(\bB')$.
\end{lemma}

\section{$\mathbb{Q}$ versus $\mathbb{Q}_p$}

\noindent
Note that the structure ${\mathfrak Q}_p$ has a substructure with domain ${\mathbb Q}$.
All our algorithms  in Section \ref{sec:algorithms} and hardness proofs in Section \ref{sec:hard} 
can be stated equivalently over the uncountable field ${\mathbb Q}_p$ or over $\mathbb{Q}$. This is possible due to the following fact, which is a consequence of a result of Weispfenning~\cite{LinearWeispfenning}.

\begin{proposition}\label{prop:qvsqp}
For every $p\in\mathbb{P}$, the structure ${\mathfrak Q}_p$ and its substructure with domain ${\mathbb Q}$ have the same first-order theory.
\end{proposition}

\begin{proof}
Let $\tau$ be the signature $\{+,\cdot,0,1,\pi,{\rm div}\}$ 
where $\pi$ is a constant symbol and ${\rm div}$
is a binary relation symbol.
{Weispfenning~\cite{LinearWeispfenning} 
introduces a certain first-order $\tau$-theory, which he calls $T_{{\rm DVF}_p}$; both $\mathbb{Q}$ and $\mathbb{Q}_p$ give rise to models of $T_{{\rm DVF}_p}$
if $\pi$ is interpreted as $p$
and $a\;{\rm div}\;b$ if and only if $v_p(a) < v_p(b)$.
He then proves that $T_{{\rm DVF}_p}$ 
admits quantifier elimination for {\em linear} formulas~\cite[Theorem 3.6]{LinearWeispfenning}.
That is, every $\sigma$-formula, for $\sigma := \{+,0,1,\pi,{\rm div}\}$ (where the symbol for multiplication is missing, which is why these formulas are called `linear'), is over $T_{{\rm DVF}_p}$ equivalent to a quantifier-free $\sigma$-formula.}
Clearly, every atomic formula in the signature of ${\bQ}_p$ can be defined by a $\sigma$-formula 
over ${\mathbb Q}_p$. Let $\varphi$ be a first-order sentence in the signature of $\bQ_p$,
and
let $\varphi'$ be the first-order $\sigma$-sentence 
obtained from $\varphi$ by replacing all atomic formulas  by their defining $\sigma$-formula. 
Then $\varphi'$ is either equivalent to $0=0$ over $T_{{\rm DVF}_p}$
or it is equivalent to $0=1$ over $T_{{\rm DVF}_p}$. It follows that
either both ${\bQ}_p$ and its substructure with domain ${\mathbb Q}$ satisfy $\varphi$,
or both ${\bQ}_p$ and its substructure with domain ${\mathbb Q}$ satisfy $\neg \varphi$, which is what we wanted to show.
\end{proof}

\begin{corollary}\label{cor:NP}
For each $p\in\mathbb{P}$,
the existential theory of $\mathfrak{Q}_p$
and
the existential theory of the expansion of $({\mathbb Q};+,1)$
by all relations of the form $\leq_c^p$, $\geq_c^p$, $=_c^p$ and $\neq_c^p$, for $c \in {\mathbb Z}$, are in NP. 
\end{corollary} 
\begin{proof}
By Proposition~\ref{prop:qvsqp},
these two existential theories are equal,
so the claim follows from
\cite[Proposition 21]{GHW19},
where it is proven that the existential theory of $\mathbb{Q}_p$ in a more expressive language is in NP.
\end{proof}

\section{Algorithms}
\label{sec:algorithms}

\noindent
We first discuss how to measure the size of input instances to the computational problems studied in this text. 
For $a,b\in\mathbb{N}\setminus \{0\}$ coprime, define $h(\pm\frac{a}{b}) := 1+\log|a|+\log|b|$, and define $h(0) := 1$.
Occasionally we might allow special coefficients like $\infty$ or $-\infty$; we set $h(\infty) =h(-\infty) :=1$. 
For matrices $A_1,\dots,A_r$ with coefficients in $\mathbb{Q}$ we let
$$
 C(A_1,\dots,A_r) := s + \sum_{k=1}^r\sum_{i,j}h(a_{kij}),
$$
where $s$ is the maximal number of rows or columns of one of the $A_k = (a_{kij})_{i,j}$. 
This is our measure of size of a computational problem that is given by a set of rational matrices.
A rational number $p$ in the input is interpreted as the matrix $p\in\mathbb{Q}^{1\times1}$,
and a finite set $D=\{d_1,\dots,d_r\}\subseteq\mathbb{Q}$ 
is interpreted
as the matrix
$D=(d_1,\dots,d_r)\in\mathbb{Q}^{1\times r}$.
For example, the input size of the algorithm in Proposition~\ref{prop:alg-leq} below is $C(A,b,p,c,D_1,\dots,D_n)$.

We now present two algorithms.
The first one, essentially for constraints of the form $v_p(x)\leq c$,
is straightforward.
In both settings, among all such valuation constraints on the same variable, there is a most restrictive one, which can easily be identified (in polynomial time), and therefore our algorithms are only formulated for one valuation constraint of the form $v_p(x)\leq c$ (or of the form $v_p(x)\geq c$) per variable.

\begin{proposition}\label{prop:alg-leq}
There is a polynomial time algorithm 
that decides,
given
$m,n\in\mathbb{N}$,
$p\in\mathbb{P}$,
$c\in(\mathbb{Z}\cup\{\infty\})^n$,
$A\in \mathbb{Q}^{m\times n}$, $b\in\mathbb{Q}^m$,
and finite sets
$D_1,\dots,D_n\subseteq\mathbb{Z}$,
whether there exists
$x\in\mathbb{Q}^n$ with
$Ax=b$ such that
$v_p(x_j)\leq c_j$ and
$v_p(x_j)\notin D_j$ for $j=1,\dots,n$.
\end{proposition}

\begin{proof}
Let $L :=\{x\in\mathbb{Q}^n:Ax=b\}$ be the solution space of the system of linear equations.
If $L=\emptyset$, the algorithm outputs NO.
Otherwise write
\begin{equation}\label{eqn:solution_space}
     L = \left\{y_0+\sum_{k=1}^d\lambda_ky_k : \lambda_1,\dots,\lambda_d\in\mathbb{Q} \right\}
\end{equation}
with $y_1,\dots,y_d\in {\mathbb Q}^{n}$ linearly independent.
One can check whether $L=\emptyset$
and otherwise compute such $d \in {\mathbb N}$ and $y_0,\dots,y_d$ in polynomial time: 
It is possible to compute 
one solution $y_0 \in {\mathbb Q}^{n}$ of $Ax=b$ in polynomial time \cite[Corollary 3.3a]{Schrijver}.
Moreover, we can transform $A$ by elementary row operations into a matrix $A'$ in row echelon form in polynomial time \cite[Theorem 3.3]{Schrijver},
and from $A'$ we can read off 
the rank $d$ of $A$
and a basis $y_1,\dots,y_d$ of 
$$\{x\in\mathbb{Q}^n:Ax=0\}=\{x\in\mathbb{Q}^n:A'x=0\}.$$

If $c_j=\infty$ let $C_j=(\mathbb{Z}\cup\{\infty\})\setminus D_j$,
otherwise
let $C_j=(-\infty,c_j]\setminus D_j$,
so that the algorithm has to decide whether there exists $x\in L$ with $v_p(x_j)\in C_j$ for every $j$.
If for some $j$ we have that  $v_p(y_{0,j})\notin C_j$ and $y_{k,j}=0$ for every $k=1,\dots,d$,
then every $x\in L$ satisfies
$v_p(x_j)=v_p(y_{0,j})\notin C_j$, and the algorithm outputs NO.
Otherwise, the algorithm outputs YES.
To see that this is the correct answer, assume now that for every $j$
we have $v_p(y_{0,j})\in C_j$ or 
$y_{k,j}\neq 0$ for some $k$.
Let $c_j'=\sup(\mathbb{Z}\setminus C_j)\in\mathbb{Z}\cup\{\infty\}$,
where we set $c_j' := \infty$ if $C_j=\mathbb{Z}\cup\{\infty\}$, 
and let
$$
 e := \max\{|v_p(y_{k,j})|:k=0,\dots,d;j=1,\dots,n;y_{k,j}\neq0\}+\max\{0,-c'_1,\dots,-c'_n\}+1.
$$ 
We claim that 
$$
 x:=y_0+\sum_{k=1}^d p^{-2ke}y_k
$$ 
is a solution to all the constraints.
For each $j$
let 
$$
 K_j=\{k\in\{0,\dots,d\}:y_{k,j}\neq 0\}.
$$ 
If $K_j\setminus\{0\}=\emptyset$, then, by our assumption,
$v_p(x_j) =v_p(y_{0,j})\in C_j$.
Otherwise,
$$
 -e(2k+1) = -2ke-e< v_p(p^{-2ke}y_{k,j}) <-2ke+e = -e (2k-1)
$$ 
for every $k\in K_j$,
so that the $v_p(p^{-2ke}y_{k,j})$ for $k\in K_j$
are pairwise distinct,
and therefore,
with $k_j:=\max K_j$,
\begin{align*} v_p(x_j) & =v_p\Big(\sum_{k=0}^{k_j}p^{-2ke}y_{k,j}\Big)=-2k_je+v_p(y_{k_j,j})<c'_j
\end{align*}
by the choice of $e$. 
This shows in particular that $v_p(x_j)\in C_j$,
as required.
\end{proof}

\begin{remark}
    We might not be able to compute a solution in the usual binary representation, as already for the single  constraint $v_p(x)\leq c$ the smallest solution (with respect to the $p$-adic absolute value $|x|_p := p^{-v_p(x)}$) 
    is $p^{-c}$. 
The algorithm not only works for 
the $p$-adic valuation on $\mathbb{Q}$
but for arbitrary so-called discrete valuations 
on a computable field $K$ in which a solution of a given linear equation,
a basis for the solution space of a homogeneous linear equation,
and the valuation of an element
can be computed; the resulting algorithm has a  polynomial running time if these computations can be performed in polynomial time.
\end{remark}

For our second algorithm we need some preparations.
As the algorithm achieves a stronger result, we just mention without proof that the usual Hermite normal form allows to check in polynomial time whether $Ax=b$ has a solution $x\in\mathbb{Z}_{(p)}^n$ (see, e.g.,~\cite[{Chapter 5}]{Schrijver}). 
However, already checking for solutions $x$ with $x_j\in\mathbb{Z}_{(p)}$ for $1\leq j\leq r$ and $x_j\in\mathbb{Q}$ for $r+1\leq j\leq n$ requires new ideas.
Also,
if we want to allow constraints 
of the form $v_p(x_j)\geq c_j$ rather than just $v_p(x_j)\geq 0$,
one could replace $x_j$ by $x_jp^{-c_j}$, 
but only as long as $p^{c_j}$ is polynomial in the input size.
This would be the case if the $c_j$ would be coded in unary,
but if the $c_j$ are coded in binary, as is our convention (see above),
replacing $x_j$ by $x_jp^{-c_j}$ will blow up the coefficients of the linear equation exponentially.
We therefore do {\em not} replace $x_j$ by $x_jp^{-c_j}$ but instead do some extra bookkeeping, exploiting the fact that although we might not be able to compute finite sums of elements of the form
$x_jp^{c_j}$ in polynomial time,
we can at least compute their value. 

\begin{lemma}\label{lem:computeval}
There is a polynomial-time algorithm which,
given $p\in\mathbb{P}$, $n\in\mathbb{N}$, and pairs 
$(a_1,c_1),\dots,$
$(a_n,c_n)\in\mathbb{Q}\times\mathbb{Z}$, computes
$v_p(\sum_{i=1}^na_ip^{c_i})\in\mathbb{Z}\cup\{\infty\}$.
\end{lemma}

\begin{proof}
First remove all $(a_i,c_i)$ with $a_i=0$ from the list.
If $n=0$ then output $\infty$.
Replace each $(a_i,c_i)$ by $(a_ip^{-v_p(a_i)},c_i+v_p(a_i))$
to assume that $v_p(a_i)=0$.
Let $c=\min_ic_i$. If there exists a unique $i_0$ with $c=c_{i_0}$, then output $v_p(\sum_ia_ip^{c_i})=c$.
Otherwise assume without loss of generality that $c=c_1=c_2$.
Then $a_1p^{c_1}+a_2p^{c_2}=(a_1+a_2)p^c$.
Remove $(a_1,c_1)$ and $(a_2,c_2)$ from the list and append $(a_1+a_2,c)$.
Repeating this process will terminate after at most $n$ steps.
\end{proof}

We also need a certain 
row echelon form. Before we give the definition, we present two motivating examples. 

\begin{example}
Suppose we want check whether a linear equation
\begin{align}
    a_1x_1 + \dots + a_nx_n = b  \label{eq:lin}
\end{align}
with $a_1,\dots,a_n,b\in\mathbb{Q}$
has a solution $x\in\mathbb{Q}^n$ with $v_p(x_j)\geq 0$ for every $j \in \{1,\dots,n\}$.
Such an $x$ exists if and only if $v_p(b) \geq \min_jv_p(a_j)$:
For any such $x$, 
\begin{align*}
v_p(b) = 
v_p(a_1x_1+\dots+a_nx_n) & \geq\min\{v_p(a_1)+v_p(x_1),\dots,v_p(a_n)+v_p(x_n)\}\\
& \geq \min_j v_p(a_j),
\end{align*}
and conversely, if $v_p(a_{j_0}) \leq v_p(b)$ for some $j_0$, we can let $x_{j_0} := a_{j_0}^{-1}b$ and set the other $x_j$ to $0$
(unless $a_{j_0}=0$, in which case $b=0$ and we can let $x=0$).
\end{example}

\begin{example}\label{expl:lineq2}
Suppose we are 
given a nonempty set $X\subseteq\mathbb{Z}_{(p)}^{n-1}$
and want to check whether for some $(x_2,\dots,x_n)\in X$ there exists $x_1 \in \mathbb{Z}_{(p)}$ satisfying~\eqref{eq:lin}. 
As long as $a_1\neq 0$, we can solve for $x_1$ and obtain
$$
 x_1 = a_1^{-1}\big( b - \sum_{j=2}^na_jx_j\big).
$$
However, computing
$v_p(x_1)$
can be difficult from just the values  $v_p(x_j)$ for $j=2,\dots,n$, since 
we are only guaranteed
$$
 v_p\Big(a_1^{-1}(b-\sum_{j=2}^na_jx_j)\Big) \geq \min\{ v_p(b)-v_p(a_1), 
 \min_{j=2,\dots,n}(v_p(a_j)-v_p(a_1)+v_p(x_j)) \}
$$
and it can happen that the inequality is strict. The right hand side is certainly nonnegative as long as
$v_p(a_1)\leq v_p(b)$
and $v_p(a_1)\leq v_p(a_j)$ for every $j$.
And in fact,
when $v_p(a_1)\leq v_p(a_j)$  for every $j$, the condition $v_p(a_1)\leq v_p(b)$ is also necessary
for the left hand side to be nonnegative:
If $v_p(a_1)>v_p(b)$, then
$v_p(b)-v_p(a_1)<0$ but $v_p(a_j)-v_p(a_1)+v_p(x_j)\geq0$ for every $j$, so the inequality is actually an equality.
Therefore, as long as $a_1$ has minimal valuation among the $a_i$, for any $(x_2,\dots,x_n)\in X$
there exists $x_1 \in \mathbb{Z}_{(p)}$ satisfying~\eqref{eq:lin} if and only if 
$v_p(a_1) \leq v_p(b)$.
This criterion easily generalizes to systems of several equations $Ax=b$ where $A$ 
is in row echelon form and 
each pivot element has minimal valuation in its row.
This is what Definition~\ref{def:fHNF} below expresses in the special case of the function $f(a,j)=v_p(a)$.
\end{example}

\begin{definition}\label{def:fHNF}
A \emph{pivot function} is a function 
$$f \colon \mathbb{Q} \times {\mathbb N}  \rightarrow\mathbb{Q}\cup\{\infty,-\infty\}$$
such that $f(a,j)=\infty$ 
if and only if $a=0$.
For a {pivot function} $f$,
we say that a matrix $A=(a_{ij})_{i,j} \in {\mathbb Q}^{m \times n}$ 
is in \emph{$f$-minimal row echelon form} if the following two conditions are satisfied. 
\begin{enumerate}[(a)]
    \item $A$ is in row echelon form, i.e., setting $j_i:=\inf\{j:a_{ij}\neq0\}$ {for $i \in \{1,\dots,n\}$,} 
    there exists $k \in \{0,\dots,m\}$ such that $j_1<\dots<j_k<j_{k+1}=\dots=j_m=\infty$. 
    \item Each pivot element $a_{i,j_i}$ of $A$ minimizes $f$ within its row in the sense that for each $i\in\{1,\dots,k\}$,
    $$
     f(a_{ij_i},j_i) = \min\{ f(a_{ij},j) : j=j_i,\dots,n\}.
    $$
\end{enumerate}
\end{definition}

\begin{example}
To explain why we need more general functions $f$ than 
just $f(a,j) = v_p(a)$, suppose we replace the conditions $v_p(x_j)\geq 0$ 
in Example~\ref{expl:lineq2} 
by $v_p(x_j)\geq c_j$ for some $c_j$.
Rewriting this as $v_p(x_jp^{-c_j})\geq 0$ we see that we could instead consider the matrix
$A'=(a'_{ij})_{i,j}$ given by $a_{ij}'=a_{ij}p^{c_j}$ and apply the criterion from Example~\ref{expl:lineq2}. However, the numbers $p^{c_j}$ have exponential representation size. This can be avoided by replacing the condition that
each pivot element $a_{i{j_i}}p^{c_{{j_i}}}$ of $A'$ minimizes the function $v_p$ within its row
by the condition that 
each pivot element
$a_{ij_i}$ of $A$ minimizes the function $f(a_{ij},j) := v_p(a_{ij}) +c_j$ within its row,
where the second argument indicates the column.
\end{example}

We write $\GL_m({\mathbb Q})$ for the general linear group of degree $m$ over the field ${\mathbb Q}$, i.e., the group
of all invertible matrices in ${\mathbb Q}^{m \times m}$.
If $\sigma\in S_n$ is a permutation, then $P_\sigma=(\delta_{i,\sigma(i)})_{i,j}\in\GL_n(\mathbb{Q})$
denotes the corresponding permutation matrix.
For a pivot function $f$ and
$\sigma \in S_n$, we write 
 $f_\sigma$ for the pivot function given by
$$f_\sigma(a,j) := \begin{cases}
    f(a,\sigma^{-1}(j)) & \text{ if } j \in \{1,\dots,n\} \\ f(a,j) & \text{ otherwise.} 
\end{cases}
$$

If $S$ is a set, then $S^*$ denotes the set of non-empty words over the alphabet $S$, i.e., 
the set of finite sequences of elements of $S$.

\begin{lemma}\label{lem:compp-echelon}
Let $f \colon \mathbb{Q}\times\mathbb{N}\times (\mathbb{Z} \cup \{-\infty\})^*\rightarrow\mathbb{Q}\cup\{\infty,-\infty\}$
and assume that for each
$c\in(\mathbb{Z} \cup \{-\infty\})^*$,
the map $f_c$ defined by $(a,j)\mapsto f(a,j,c)$ is a pivot function.
For 
every
$m,n\in\mathbb{N}$, $A\in\mathbb{Q}^{m\times n}$, and $c\in (\mathbb{Z} \cup \{-\infty\})^*$
there exist
$U\in{\rm GL}_m(\mathbb{Q})$ and 
$\sigma \in S_n$ 
such that $UAP_{\sigma}$ is in
$(f_c)_\sigma$-minimal row echelon form. 
If $f$ is computable in polynomial time, then 
such
$U$ and $P_\sigma$ can be computed in polynomial time.
\end{lemma}

\begin{proof}
    We describe how to get $U$ and $P_{{\sigma}}$ in terms of elementary row and column operations, where the only elementary column operations allowed are swapping two columns. 
If $A=0$, then we are done.
Otherwise, possibly swap two rows to assume that $a_{1j}\neq 0$ for some $j$. Choose $k \in \{1,\dots,n\}$ such that
$$
 f_c(a_{1k},k)=\min\{ f_c(a_{1j},j) : j=1,\dots,n \}
$$
(which implies in particular that
$a_{1k}\neq 0$, since $f_c(0,k)=\infty$ by assumption).
If $k\neq 1$, then swap the first column with the $k$-th column.
Add multiples of the first row to the other rows to achieve that
$a_{i1}=0$ for every $i>1$. Reduce the fractions in the entries of the matrix. 
Now take the $(m-1)\times(n-1)$-submatrix with rows $i=2,\dots,m$ and columns $j=2,\dots,n$, and iterate
(extending each of the following row and column operations to the whole matrix). 
It is well-known that the representation size of the involved numbers stays polynomial (see, e.g., \cite[Theorem 3.3]{Schrijver}).
This process terminates after at most $\max\{m,n\}$ steps, and the resulting matrix is of the desired form. 
\end{proof}

In the following, if $x \in {\mathbb Q}^n$ and $c \in ({\mathbb Z} \cup \{-\infty\})^n$, we will write $v_p(x) \geq c$ if $v_p(x_j) \geq c_j$ for every $j \in \{1,\dots,n\}$. 

\begin{remark}\label{rem:transform}
Note
that if $B = U A P_\sigma$ for some $U \in {\rm GL}_m(\mathbb{Q})$ and $\sigma \in S_n$, then 
$Ax=b$ has a solution $x \in {\mathbb Q}^n$ such that 
$v_p(x) \geq c$ if and only if
$B y = U b$ has a solution $y \in {\mathbb Q}^n$ such that $v_p(y) \geq P_\sigma^{-1}c$ (the map $x \mapsto P^{-1}_\sigma x$ is a bijection between the solutions to the first system 
and the solutions to the second system).
\end{remark}

The following result allows constraints of the form $v_p(x)\geq c$ and, in the case $p=2$, constraints of the form $v_2(x)=c$. The tuple $\delta$ encodes which constraint applies to which variable. 

\begin{theorem}\label{thm:alg-geq}
There is a polynomial-time algorithm that decides, given
$m,n\in\mathbb{N}$, $p\in\mathbb{P}$, $c\in(\mathbb{Z} \cup \{-\infty\})^n$,
$\delta\in\{0,1\}^n$,
$A\in {\mathbb Q}^{m\times n}$, 
and $b\in\mathbb{Q}^m$,
whether there exists $x\in\mathbb{Q}^n$ with $Ax=b$
such that $v_p(x)\geq c$ and, in the case $p=2$, $\delta_j=1$,
and $c_j\neq-\infty$, also
$v_p(x_j)=c_j$.
\end{theorem}
\begin{proof}
We can assume that if $\delta_j=1$ for some $j$, then $p=2$ and $c_j\neq-\infty$.
Define the pivot function 
$$
f(a,j):=v_p(a)+c_j+\frac{\delta_j}{2},
$$ 
where we use the convention $\infty+(-\infty):=\infty$. 
Clearly, $f$ is computable in polynomial time (as a function of $a$, $j$, $p$ and $c$ and $\delta$).  
By Lemma \ref{lem:compp-echelon} we can 
compute $U$ and $P_\sigma$ in polynomial time such that
$UAP_\sigma$ is in 
$f_\sigma$-minimal
row echelon form. 
We may 
replace $A$ by $UAP_{\sigma}$, 
$b$ by $Ub$,
$c$ by $P_\sigma^{-1}c$, 
and $\delta$ by $P_\sigma^{-1}\delta$ (adapting the idea from Remark~\ref{rem:transform} appropriately in the case $p=2$),
and henceforth 
assume without loss of generality that
$\sigma={\rm id}$ and that 
$A$ is already in $f$-minimal row echelon form.

Let $k$ be as in Definition~\ref{def:fHNF}. 
    Note that condition (b) in Definition~\ref{def:fHNF} states that for every $i\leq k$ 
\begin{align}
v_p(a_{ij_i})+c_{j_i}+\frac{\delta_{j_i}}{2} =\min\left\{v_p(a_{ij})+c_j+\frac{\delta_j}{2} : j=j_i,\dots,n\right\}. \label{eq:nf}
    \end{align}
Since for every $i \in \{1,\dots,k\}$ and $j \in \{1,\dots,n\}$ we have $v_p(a_{ij}) \in {\mathbb Z} \cup\{\infty\}$,  $c_j\in\mathbb{Z} \cup\{-\infty\}$, 
 and $\delta_j\in\{0,1\}$,
\eqref{eq:nf} implies that 
\begin{enumerate}
    \item[$(b')$] \label{cond:b'}   $v_p(a_{ij_i})+c_{j_i}=\min\left\{v_p(a_{ij})+c_j:j=j_i,\dots,n\right\}
     $, and
    \item[$(b'')$] \label{cond:b''}$v_p(a_{ij_i})+c_{j_i}+{\delta_{j_i}}=\min\left\{v_p(a_{ij})+c_j+{\delta_j}:j=j_i,\dots,n\right\}$.
\end{enumerate}
The algorithm then outputs YES
if 
\begin{enumerate}
    \item\label{cond1} $b_i=0$ for every $i \in \{k+1,\dots,m\}$, and
    \item\label{cond2} $v_p(a_{ij_i})+c_{j_i}+\delta_{j_i}\leq v_p(b_i-\sum_{j\geq j_i}\delta_j a_{ij}p^{c_j})$ for every $i\in\{1,\dots,k\}$,
\end{enumerate}
and otherwise it outputs NO.
Note that Condition \eqref{cond2} can be checked in polynomial time by Lemma~\ref{lem:computeval}.

To see that this is the correct answer, we have to show that 
\eqref{cond1} and \eqref{cond2} holds if and only if there exists $x \in {\mathbb Q}^n$ with $Ax = b$, $v_p(x) \geq c$ and $v_p(x_j) = c_j$ for every $j$ with $\delta_j=1$.

To prove the backwards direction, assume such an $x \in {\mathbb Q}^n$ exists. Then clearly \eqref{cond1} holds, and we will argue that \eqref{cond2} must be satisfied as well. 
Suppose for contradiction that
\begin{equation}\label{eqn:nicht-2}
v_p(a_{ij_i})+c_{j_i}+\delta_{j_i}> v_p\Big(b_i-\sum_{j\geq j_i}\delta_ja_{ij}p^{c_j}\Big)    
\end{equation}
for some $i\leq k$.
Since $Ax=b$ and $A$ is in row echelon form, we have that 
\begin{align}
    x_{j_i} = a_{ij_i}^{-1} \cdot \Big ( b_i - \sum_{j > j_i} a_{ij} x_j \Big),
    \label{eq:lineq}
\end{align}
and this implies by Lemma \ref{lem:vpq} that
\begin{eqnarray}\label{eqn:new}
v_p(x_{j_i}-\delta_{j_i}p^{c_{j_i}})
&{=}& v_p\Big(a_{ij_i}^{-1} \big(b_i - \delta_{j_i}a_{ij_i}p^{c_{j_i}} 
-
\sum_{j > j_i} a_{ij} x_j 
\big) \Big)
\nonumber 
\\
&=& 
v_p\Big(
b_i -\sum_{j\geq j_i}\delta_ja_{ij}p^{c_j}
-
\sum_{j > j_i} a_{ij} (x_j- 
\delta_jp^{c_j}) \Big)
- v_p(a_{ij_i})
\nonumber 
\\
&\geq &\min\Big\{v_p\Big(b_i-\sum_{j\geq j_i}\delta_ja_{ij}p^{c_j}\Big),\min_{j>j_i}\big(v_p(a_{ij})+v_p(x_j-\delta_jp^{c_j})\big)\Big\}-v_p(a_{ij_i}).   \label{eq:nmins1}
\end{eqnarray}
Note that 
\begin{equation}\label{eqn:minusdelta}
 v_p(x_j-\delta_jp^{c_j})\geq c_j+\delta_j 
\end{equation}
for every $j$, because if $\delta_j=1$ we are in the case $p=2$ 
and have $v_p(x_j)=c_j=v_p(p^{c_j})$, 
and thus $v_p(x_j - \delta_j p^{c_j}) > c_j$ (Lemma \ref{lem:ac}); for $\delta_j=0$ the statement $v_p(x_j)\geq c_j$ holds by assumption. 
We get for every $j\geq j_i$ that
\begin{align}
 v_p\Big(b_i-\sum_{j\geq j_i}\delta_ja_{ij}p^{c_j}\Big) 
 & \stackrel{(\ref{eqn:nicht-2})}{<}
 v_p(a_{ij_i})+c_{j_i}+\delta_{j_i} 
  \nonumber 
  \\
 & \stackrel{(b'')}{\leq} 
v_p(a_{ij})+c_{j}+\delta_j \stackrel{(\ref{eqn:minusdelta})}{\leq}
v_p(a_{ij})+v_p(x_j-\delta_jp^{c_j}). 
\label{eq:ineq}
\end{align}
Therefore, by Lemma \ref{lem:vpq} the inequality in (\ref{eqn:new}) is an equality and
$$
 v_p(x_{j_i}-\delta_{j_i}p^{c_{j_i}}) 
 = v_p\Big(b_i-\sum_{j\geq j_i}\delta_j a_{ij}p^{c_j}\Big)- v_p(a_{ij_i})\\
 \stackrel{(\ref{eqn:nicht-2})}{<} c_{j_i}+\delta_{j_i},
$$
which is a contradiction to (\ref{eqn:minusdelta}) for $j=j_i$.

For the forward direction, we 
assume that \eqref{cond1} and \eqref{cond2} hold and construct $x$ as follows:
for each $j \in \{1,\dots,n\} \setminus \{j_1,\dots,j_k\}$, 
let $x_j := p^{c_j}$
if $c_j\in\mathbb{Z}$, and otherwise let $x_j:=0$. For 
$i = k,\dots,1$ define
$x_{j_i}$ iteratively by
\eqref{eq:lineq},
which implies that
\eqref{eqn:new} again holds.
The so constructed $x$ satisfies
$Ax=b$, and for each
$j\notin\{j_1,\dots,j_k\}$
that $v_p(x_j) \geq c_j$
and $v_p(x_j)=c_j$ if $\delta_j=1$
We prove by induction on 
$i =k,\dots,1$ 
that $v_p(x_{j_i}) \geq c_{j_i}$
and that $v_p(x_{j_i}) = c_{j_i}$ 
if $\delta_{j_i}=1$. 

We first consider the case $\delta_{j_i}=0$. 
Then $v_p(b_i-\sum_{j\geq j_i}\delta_ja_{ij}p^{c_j})-v_p(a_{ij_i})\geq c_{j_i}$ by \eqref{cond2}. 
Moreover, $v_p(x_j)\geq c_j$ for each $j>j_i$ by the inductive assumption,
and since also $v_p(\delta_jp^{c_j})\geq c_j$, it follows that $v_p(x_j-\delta_jp^{c_j})\geq c_j$.
Thus 
\begin{align*}
    v_p(x_{j_i}) & 
    \geq {\min\big\{c_{j_i},\min_{j > j_i} (v_p(a_{ij})  + c_j) - v_p(a_{ij_i})\big\}} && \text{(by~\eqref{eq:nmins1} and the above)} \\
    & = c_{j_i} && \text{(by $(b')$)}
\end{align*} 
as claimed.

In the case
$\delta_{j_i}=1$ we necessarily have $p=2$,
and now \eqref{cond2} gives that
$$
 v_p\Big(b_i-\sum_{j\geq j_i}\delta_ja_{ij}p^{c_j}\Big)-v_p(a_{ij_i})\geq c_{j_i}+\delta_{j_i}>c_{j_i}.
$$ 
Let $j>j_i$. We have 
$v_p(x_j)\geq c_j$ by the inductive assumption. 
If $\delta_j=0$, then $(b'')$ gives that {$v_p(a_{ij_i})+c_{j_i} +1 \leq v_p(a_{ij})+c_j$
and hence $v_p(a_{ij})+ v_p(x_j) - v_p(a_{ij_i}) > c_{j_i}$.} 
If $\delta_j=1$, then $v_p(x_j) = c_j$, which since $p=2$ implies that
$v_p(x_j-p^{c_j})>c_j$ (Lemma \ref{lem:ac}).
Now, $(b'')$ gives $v_p(a_{ij_i}) + c_{j_i} \leq v_p(a_{ij}) + c_j$, and hence $v_p(a_{ij}) + v_p(x_j-p^{c_j}) - v_p(a_{i j_i}) > c_{i_j}$.
Then $v_p(x_{j_i} - p^{c_{j_i}}) > c_{j_i}$ by~\eqref{eq:nmins1} and the above. 
Finally, 
this implies
$v_p(x_{j_i})=c_{j_i}$
since $p=2$ (Lemma \ref{lem:vpq}).
\end{proof}

\section{NP-hardness and reductions}
\label{sec:hard}

\noindent
For a set $A$ and $a \in A$, we use 
$\neq_a$ as a relation symbol for the unary relation $A \setminus \{a\}$, 
and later write $x \neq a$ instead of ${\neq_a}(x)$.

\begin{lemma}\label{lem:cyclic}
Let $G$ be a finite cyclic group of order $n\geq 3$. Then $\CSP(G;+,\neq_0)$
is NP-hard. In particular, the primitive existential theory of $(G;+)$ is NP-hard. 
\end{lemma}

\begin{proof}
The primitive positive formula
$\exists e,z (e+e = e \wedge y+z = e \wedge x+z \neq 0)$ defines the binary relation $\neq$ over  $G$. 
A finite graph with vertices $[n]$ and edges $E\subseteq [n]^2$
can be colored with $n=|G|$ colors if and only if
$\bigwedge_{(i,j)\in E}x_i\neq x_j$ is satisfiable in $G$.
For $n\geq 3$, the graph coloring problem is NP-hard \cite[{Section 4}]{GareyJohnson},
so the claim follows from Lemma~\ref{lem:pp-int}.
\end{proof}

\begin{lemma}\label{lem:interpret_ZpeZ}
For every prime number $p$ and every $e\in\mathbb{N}$ 
the structure $(\mathbb{Z}/p^e\mathbb{Z};+,\neq_0)$
has a primitive positive interpretation
in $(\mathbb{Z}_p;+,<^p_e)$.
\end{lemma}

\begin{proof}
The quotient map $\gamma \colon \mathbb{Z}_p\rightarrow\mathbb{Z}_p/p^e\mathbb{Z}_p\cong\mathbb{Z}/p^e\mathbb{Z}$ does the job:
As $\gamma^{-1}(0)=p^e\mathbb{Z}_p$ is primitively positively definable in $(\mathbb{Z}_p;+)$, 
also the pullback of the graph of $+$ is primitively positively definable in
$(\mathbb{Z}_p;+)$.
Finally, $\gamma^{-1}(\neq_0)=\mathbb{Z}_p\setminus p^e\mathbb{Z}_p=\{x\in\mathbb{Z}_p:v_p(x)<e\}$ is primitively positively definable in
$(\mathbb{Z}_p;+,<^p_e)$.
\end{proof}

\begin{proposition}\label{prop:eq0}
The primitive positive theory of
$\CSP(\mathbb{Z}_p;+,=^p_0)$ is NP-hard for $p\geq 3$, and 
$\CSP(\mathbb{Z}_p;+,\leq^p_1)$
is NP-hard for all prime numbers $p$.  
\end{proposition}

\begin{proof}
If $p \geq 3$, 
then $(\mathbb{Z}/p \mathbb{Z};+,\neq_0)$ is NP-hard by Lemma~\ref{lem:cyclic}. 
Moreover, by Lemma~\ref{lem:interpret_ZpeZ} it has a 
primitive positive interpretation in 
$({\mathbb Z}_p;+,=_0^p) = ({\mathbb Z}_p;+,<_1^p)$ 
and so
$\CSP(\mathbb{Z}_p;+,=^p_0)$ is NP-hard
by Lemma~\ref{lem:pp-int}.

If $p$ is an arbitrary prime number, then  $(\mathbb{Z}/p^2\mathbb{Z};+)$ 
is cyclic of order $p^2\geq 3$ and we have that $\CSP(\mathbb{Z}/p^2 \mathbb{Z};+,\neq_0)$ is NP-hard by Lemma~\ref{lem:cyclic}. 
The structure $(\mathbb{Z}/p^2 \mathbb{Z};+,\neq_0)$
has a primitive positive interpretation in 
$(\mathbb{Z}_p;+,<^p_2)$ by 
Lemma~\ref{lem:interpret_ZpeZ}, and hence 
$(\mathbb{Z}_p;+,<^p_2)= (\mathbb{Z}_p;+,\leq^p_1)$ is NP-hard by Lemma~\ref{lem:pp-int}.
\end{proof}

Let $c$ be a positive integer. In primitive positive formulas over structures whose signature contains $+$ and $1$, we use $cy$ as a shortcut for $\underbrace{y+\cdots+y}_{c \text{ times}}$,
and $c$ as a shortcut for $c1$. 
We also freely use the term $x+c$ for $c \in {\mathbb Z}$; if $c = 0$, then this can be replaced by $x$, and if $c < 0$, then 
this can be rewritten into a proper primitive positive formula by introducing a new existentially quantified variable $y$, replacing $x+c$ by $y$, and 
adding a new conjunct $x = y + |c|$.

\begin{lemma}\label{lem:geq0}
For $p\geq 3$, the primitive positive formula 
$$
\exists y,z \big (v_p(y)=0\wedge v_p(z)=0\wedge x=y+z \big )
$$
defines the relation $\geq^p_0$ 
in $({\mathbb Q}_p;+,=^p_0)$. 
The primitive positive formula 
$$
\exists y,z \big (v_2(y)=0\wedge v_2(z)=0\wedge 2x=y+z\big )
$$ 
defines the relation $\geq^2_0$
in $({\mathbb Q}_2;+,=^2_0)$.
\end{lemma}

\begin{proof}
First let $p\geq 3$.
Suppose that $x \in {\mathbb Q}_p$ is such that $v_p(x)\geq 0$. 
Let $i_0\in\{0,\dots,p-1\}$ be 
such that 
$v_p(x-i_0)>0$ (Lemma \ref{lem:ac}).
Since $p\geq 3$, 
there exists $i\in\{1,\dots,p-1\}\setminus\{i_0\}$, 
and $x=(x-i)+i$ with $v_p(x-i)=0$ and $v_p(i)=0$.
Then setting $y$ to $x-i$ and $z$ to $i$, all the three conjuncts of the given formula are satisfied.
Conversely, if $v_p(y)=v_p(z)=0$, then $v_p(y+z)\geq 0$.

For $p=2$,
if $x \in {\mathbb Q}_2$ is such that $v_2(x)\geq 0$, then 
$2x=(2x-1)+1$ with $v_2(2x-1)=0$ and $v_2(1)=0$.
Conversely,
if $y,z \in {\mathbb Q}_2$ are such that $v_2(y)=v_2(z)=0$, then
$v_2(y+z)>0$, so if $2x=y+z$, then
$v_2(x)\geq 0$.
\end{proof}

The following solves an open problem from~\cite[Remark 23]{GHW19} for $p=3$; the NP-hardness for $p \geq 5$ was already shown in~\cite[Prop.~22]{GHW19}. 

\begin{corollary}\label{cor:v=0}
Let $p \geq 3$ be prime. Then 
$\CSP({\mathbb Q}_p;+,=^p_0)$ is NP-hard.
\end{corollary}
\begin{proof}
    Note that $({\mathbb Z}_p;+,=^p_0)$ has a primitive positive interpretation in $({\mathbb Q}_p;+,=^p_0)$,
    because $\geq_0^p$ is primitive positive definable in $({\mathbb Q}_p;+,=^p_0)$ by Lemma~\ref{lem:geq0}.
    Since 
    $\CSP({\mathbb Z}_p;+,=^p_0)$ is NP-hard by Proposition~\ref{prop:eq0}, the statement follows from 
    Lemma~\ref{lem:pp-int}. 
\end{proof}

\begin{lemma}\label{lem:eq0}
Let $c\in\mathbb{Z}$. The relation $=_c^2$  has the primitive positive definition 
$$
\exists y \big (v_2(y)\geq 0\wedge x=2^c+2^{c+1}y \big)
$$ 
in 
$({\mathbb Q}_2;+,1,\geq_0^2)$,
and in $({\mathbb Z}_2;+,1)$ the primitive positive definition 
$$
\exists y (x=2^c+2^{c+1}y ).
$$ 
\end{lemma}
\begin{proof}
If $v_2(x)=c$, then $x=2^c+2^{c+1}y$ with $v_2(y)\geq0$, i.e., $y\in\mathbb{Z}_2$ (Lemma \ref{lem:ac}).
Conversely, if $x=2^c+2^{c+1}y$ with $v_2(y)\geq0$, then 
$v_2(x)=\min\{v_2(2^c),v_2(2^{c+1}y)\}=c$.
\end{proof}

{Note that the primitive positive formula in Lemma~\ref{lem:eq0} has exponential representation size, since $2^{c+1}$ is a doubly exponentially large number. However, in all hardness proofs where we use this formula, $c$ will be a constant and hence the length of the formula will be a constant as well.}

\begin{lemma}\label{lem:neq0}
For all $p \in {\mathbb P}$, the relation $\neq_0^p$ has 
the primitive positive definition 
$$
 \bigwedge_{i=1}^{p-1}v_p(x-i)\leq 0
$$ in $({\mathbb Q}_p;+,1,\leq^p_0)$, and  in $({\mathbb Z}_p;+)$
the primitive positive definition $\exists y (py=x)$.
\end{lemma}

\begin{proof}
If $v_p(x)>0$, then $v_p(x-i)=v_p(i)=0$ for every $1\leq i<p$,
and if $v_p(x)<0$, then
$v_p(x-i)=v_p(x)<0$ for every $i$.
Conversely, if $v_p(x)=0$
there exists 
$i_0\in\{1,\dots,p-1\}$
with
$v_p(x-i_0)>0$
(Lemma \ref{lem:ac}).
In $\mathbb{Z}_p$, $v_p(x)\neq0$ just means $v_p(x)\geq 1$, i.e., $x=py$ with $y\in\mathbb{Z}_p$.
\end{proof}

\begin{lemma}\label{lem:leq0}
Let $d \in\mathbb{Z}$.
Then $\leq^p_d$ has the primitive positive definition 
$$
 \bigwedge_{i=1}^{p-1}v_p(x+i p^{d+1}) \neq d+1
$$ 
in $({\mathbb Q}_p;+,1,\neq^p_{d+1})$ for $p\geq 3$,
and in $(\mathbb{Q}_2;+,1,\neq^2_d)$
the primitive positive definition 
$$
 v_2(x+2^d)\neq d.
$$ 
\end{lemma}

\begin{proof}
First let $p\geq 3$.
If $v_p(x)\leq d$,
then $v_p(x+i p^{d+1})=v_p(x) < d+1$ for every $i=1,\dots, p-1$.
Conversely, if $v_p(x)>d$,
then either $v_p(x)>d+1$, 
in which case $v_p(x+i p^{d+1})=d+1$ for every $i=1,\dots,p-1$,
or $v_p(x)=d+1$.
In this case, there exists (exactly) one $i_0 \in \{1,\dots,p-1\}$ with
$v_p(x+i_0 p^{d+1})>d+1$
(Lemma \ref{lem:ac}),
and $v_p(x+i p^{d+1})=v_p(p^{d+1})=d+1$ for all $i \in \{1,\dots,p-1\} \setminus \{i_0\}$.
Such an $i$ exists by the assumption that $p\geq 3$.

{Now let $p=2$.
If $v_2(x)<d$,
then $v_2(x+2^d)=v_2(x)<d$,
and if $v_2(x)=d$,
then $v_2(x+2^d)>d$ (Lemma \ref{lem:ac}).
Conversely,
if $v_2(x)>d$, then $v_2(x+2^d)=d$.} 
\end{proof}

\begin{theorem}\label{thm:class}
    Let $p \in {\mathbb P}$ be such that $p \geq 3$. Let ${\mathfrak R}$ be a reduct of ${\mathfrak Q}_p$ whose signature $\tau$ contains $\{+,1\}$. 
    Then $\CSP(\bR)$ is in $P$ if 
    $\bR$ is a reduct of one of the structures
    \begin{align} 
    & ({\mathbb Q}_p;+,1,(\leq_c^p)_{c \in {\mathbb Z}},(\neq_c^p)_{c \in {\mathbb Z}}) \label{eq:red1} \\
    & ({\mathbb Q}_p;+,1,(\geq_c^p)_{c \in {\mathbb Z}}), \label{eq:red2}
    \end{align}
    and is NP-complete otherwise. 
\end{theorem}
\begin{proof}
The containment of $\CSP(\bR)$ in NP follows from Corollary \ref{cor:NP}.
    If $\tau$ contains $=_c^p$ for some $c \in {\mathbb Z}$, then the relation
    $=_0^p$ is primitively positively definable in $\bR$ and $\CSP(\bR)$ is NP-hard by Corollary~\ref{cor:v=0} and Lemma~\ref{lem:pp-int}. 
    So suppose that $\tau$ does not contain $=_c^p$ for any $c \in {\mathbb Z}$. 
    If $\bR$ does not contain $\geq_c^p$ for any $c \in {\mathbb Z}$, then $\bR$ is a reduct of the structure in~\eqref{eq:red1}. In this case, the polynomial-time tractability of $\CSP(\bR)$ follows from 
    Proposition \ref{prop:alg-leq}
    and Proposition \ref{prop:qvsqp}.
    So suppose that $\bR$ contains 
    $\geq_c^p$ for some $c \in {\mathbb Z}$. 
    If $\tau$ also contains $\leq_d^p$ 
    for some $d \in {\mathbb Z}$, then the relation $=_0^p$ is primitively positively definable as well, and we are again done.
    If $\tau$ contains $\neq_c^p$ for some $c \in {\mathbb Z}$, then $\leq_{c-1}^p$ is primitively positively definable in $\bR$ by Lemma~\ref{lem:leq0}, and we are in a case that we have already treated. 
    Otherwise, $\tau$ contains neither of  
    $\neq_c^p$, $\leq_c^p$, and $=_c^p$ for any $c \in {\mathbb Z}$, and hence $\bR$ is a reduct of the structure~\eqref{eq:red2}. 
    The polynomial-time tractability in this case follows from 
    Theorem~\ref{thm:alg-geq} and Proposition \ref{prop:qvsqp}.  
\end{proof}

 \begin{theorem}\label{thm:class2}
   Let $\bR$ be a reduct of ${\mathfrak Q}_2$
   whose signature $\tau$ contains $\{+,1\}$. 
   Then $\CSP(\bR)$ is in $P$ if 
    $\bR$ is a reduct of one of the structures 
    \begin{align} 
    & ({\mathbb Q}_2;+,1,(\leq_c^2)_{c \in {\mathbb Z}},(\neq_c^2)_{c \in {\mathbb Z}}) \label{eq:red3} \\
    & ({\mathbb Q}_2;+,1,(=^2_c)_{c \in {\mathbb Z}},(\geq_c^2)_{c \in {\mathbb Z}}), \label{eq:red4} 
    \end{align}
    and is NP-complete otherwise. 
\end{theorem}
\begin{proof}
    The containment of $\CSP(\bR)$ in NP follows again from Corollary \ref{cor:NP}.
    If $\tau$ contains neither $\neq_c^2$ nor $\leq_c^2$ for any $c \in {\mathbb Z}$,
    then $\bR$ is a reduct of the structure in~\eqref{eq:red4}, and the polynomial-time tractability of $\CSP(\bR)$ follows from Theorem~\ref{thm:alg-geq} and Proposition \ref{prop:qvsqp}. 
    Otherwise, the relation $\leq_1^2$ is primitively positively definable in $\bR$ by Lemma~\ref{lem:leq0}. 
    If additionally $\geq_0^2$ is primitively positively definable in $\bR$, then
    the structure $({\mathbb Z}_2;+,\leq_1^2)$ has a primitive positive interpretation in $\bR$, and the NP-hardness of $\CSP(\bR)$ follows from Proposition~\ref{prop:eq0} via Lemma~\ref{lem:pp-int}. 
    If not, then by Lemma~\ref{lem:geq0} we may assume that $\tau$ contains neither $\geq_c^p$ nor $=_c^p$ for any $c \in {\mathbb Z}$. In this case, 
    $\bR$ is a reduct of the structure in~\eqref{eq:red3}, and the polynomial-time tractability of $\CSP(\bR)$ follows from Proposition~\ref{prop:alg-leq} and Proposition \ref{prop:qvsqp}. 
\end{proof}

\section{Combining several primes, and the ordering}

\noindent
The complexity classification results for reducts of $\bQ_p$ from Theorems~\ref{thm:class} and~\ref{thm:class2} translate 
to complexity classification results 
for expansions of $({\mathbb Q};+,1)$ by 
relations from 
$$
\tau_p := \{\leq_c^p,\geq_c^p,=_c^p,\neq^p_c \; \mid c \in {\mathbb Z}\},
$$ 
for fixed $p \in {\mathbb P}$, via Proposition~\ref{prop:qvsqp}. 
Interestingly, we can even derive results 
about expansions of $({\mathbb Q};+,1)$
by relations from $\bigcup_{p \in {\mathbb P}} \tau_p$. 
Moreover, we may also obtain results about expansions of $({\mathbb Q};+,1,<)$
and of $({\mathbb Q};+,1,\leq)$.
The key to this is the following consequence of the approximation theorem for absolute values. 
As in the introduction, define $|x|_p := p^{-v_p(x)}$ for $x \in {\mathbb Q}$.

\begin{lemma}\label{lem:approx}
Let $m,n,r\in\mathbb{N}$, $\epsilon>0$, $A\in {\mathbb Q}^{m\times n}$, $b\in\mathbb{Q}^m$,
and let $p_1,\dots,p_r$ be distinct prime numbers.
For each $i\in\{0,\dots,r\}$ let $x^{(i)}\in\mathbb{Q}^n$ be such that  $Ax^{(i)}=b$.
Then there exists $x\in\mathbb{Q}^n$ with $Ax=b$ such that for every $j \in \{1,\dots,n\}$ 
and $i \in \{1,\dots,r\}$ we have 
$|x_j-x^{(0)}_j|<\epsilon$ 
and
$|x_j-x^{(i)}_j|_{p_i}<\epsilon$.
\end{lemma}

\begin{proof}
Write the solution space $L\subseteq\mathbb{Q}^n$ of $Ax=b$ as in (\ref{eqn:solution_space}).   
The map $L\rightarrow\mathbb{Q}^d$, $y_0+\sum_{k=1}^d\lambda_ky_k\mapsto (\lambda_1,\dots,\lambda_d)$
is a homeomorphism with respect to the real topology and with respect to each $p$-adic topology.
We can therefore assume without loss of generality that
$L=\mathbb{Q}^n$, i.e., that $A=0$ and $b=0$.
The claim is then precisely
the statement of the approximation theorem for finitely many inequivalent absolute values on a field $K$ (\cite[Ch.~XII, Thm.~1.2]{Lang}) in the case $K=\mathbb{Q}$,
applied for each $j\in\{1,\dots,n\}$.
\end{proof}

Let $\mathfrak Q$ be the expansion of $({\mathbb Q};+,1)$ by new relations
for the symbols from 
$$
 \tau := \{<\} \cup \bigcup_{p \in \mathbb P} \tau_p. 
$$  

\begin{proposition}\label{prop:combine} 
Let $\varphi$ be a conjunction of atomic 
$(\{+,1\} \cup \tau)$-formulas. 
Let $\varphi_<$ be all conjuncts of $\varphi$ formed with the symbol $<$, 
let $\varphi_p$
be all conjuncts of $\varphi$ formed with symbols from $\tau_p$, 
and let $\varphi_=$ be all the  conjuncts formed with $=$. 
Then $\varphi$ is satisfiable in $\mathfrak Q$ if and only if 
$\varphi_= \wedge \varphi_<$ is satisfiable in $\mathfrak Q$
and $\varphi_= \wedge \varphi_p$ is satisfiable in $\mathfrak Q$
for each $p \in {\mathbb P}$.
\end{proposition}

\begin{proof}
The forward implication is trivial. 
For the converse,  
let $s^< \in {\mathbb Q}^n$ be 
a satisfying assignment for 
$\varphi_= \wedge \varphi_<$,
let $P$ denote the (finite) set of prime numbers such that $\varphi$ contains symbols from $\tau_p$,
and for each $p\in P$ let
$s^{(p)} \in {\mathbb Q}^n$ be a satisfying assignment for 
$\varphi_= \wedge \varphi_p$.
The set $U_<\subseteq\mathbb{Q}^n$ of satisfying assignments for $\varphi_<$
is open in the real topology, 
and the set $U_p$ of satisfying assignments for $\varphi_p$
is open in the $p$-adic topology, for each $p$.
In particular, there exists $\epsilon>0$ such that 
the whole box $\{y\in\mathbb{Q}^n:|y_j-s_j^<|<\epsilon\mbox{ for every $j$}\}$ is contained in $U_<$,
and similarly $\{y\in\mathbb{Q}^n:|y_j-s_j^{(p)}|_p<\epsilon\mbox{ for every $j$}\}\subseteq U_p$ for every $p\in P$.
Therefore, by Lemma~\ref{lem:approx}, there exists 
$s \in {\mathbb Q}^n$ such that 
$s$ satisfies $\varphi_=$
and $s\in U_<\cap\bigcap_{p\in P}U_p$,
hence $s$ is a satisfying assignment for $\varphi$.
\end{proof}

Proposition~\ref{prop:combine} only works for strict inequalities, and the corresponding statement would be false for weak inequalities. On the algorithmic side, however, there is a way to reduce the problem to the satisfiability problem for strict inequalities, and we obtain the following result.

\begin{theorem}\label{thm:combine}\label{thm:inequalities}
Let $\bR$ be a reduct of $(\bQ,\leq)$ whose signature  contains $\{1,+\}$. 
If the signature of $\bR$ contains 
\begin{itemize}
    \item $=^p_c$ for some $c \in {\mathbb Z}$ and $p \in {\mathbb P}$ with $p \geq3$, or 
    \item $\geq_{c_1}^p$ and a relation from $\{\leq_{c_2}^p,\neq_{c_2}^p\}$  
    for some $c_1,c_2 \in {\mathbb Z}$ and $p \in {\mathbb P}$ with $p \geq3$, 
    \item a relation from $\{\geq_{c_1}^2,=_{c_1}^2\}$ and a relation from $\{\leq_{c_2}^2,\neq_{c_2}^p\}$ 
    for some $c_1,c_2 \in {\mathbb Z}$, 
\end{itemize}
then $\CSP(\bR)$ is NP-complete; otherwise, $\CSP(\bR)$ 
is in P.
\end{theorem}

\begin{proof}
If for some $p \geq 3$, the signature of $\bR$ contains a symbol of the form $=_c^p$,
or a relation of the form $\geq_c^p$ and a symbol of the form $\leq_c^p$ or $\neq_c^p$, 
then the NP-hardness of $\CSP(\bR)$ follows from Theorem~\ref{thm:class}
and Proposition \ref{prop:qvsqp}. 
Moreover, if 
the signature contains a symbol of the form $=^2_c$ or $\geq^2_c$ and a symbol of the form $\leq^2_p$ or $\neq_c^p$, then the NP-hardness of $\CSP(\bR)$ follows from Theorem~\ref{thm:class2}
and Proposition \ref{prop:qvsqp}.

Otherwise, let $\varphi$ be an instance of $\CSP(\bR)$. 
Similar to Proposition \ref{prop:combine}  let 
\begin{itemize}
    \item 
$\varphi_<$ be the conjuncts of $\varphi$ formed with the symbol $<$,
\item 
$\varphi_\leq$ the conjuncts formed with $\leq$,
\item 
$\varphi_p$ the conjuncts formed with symbols from $\tau_p$,
and
\item $\varphi_=$ the conjuncts formed with $=$.
\end{itemize}
Let $P$ be the set of $p\in\mathbb{P}$ such that a symbol from $\tau_p$ occurs in $\varphi$.
For any instance $\psi$ denote by $\psi^<$ the instance obtained by replacing all $\leq$ by $<$.

We first check with known methods whether there is a solution for $\varphi_0:=\varphi_=\wedge\varphi_<\wedge\varphi_\leq$
(see, e.g.,~\cite[final remark in Section 13.4]{Schrijver}). 
If there is no solution, then output NO.
Otherwise,
let $\Psi$ be the set of conjuncts of $\varphi_\leq$.
We then test for each $\psi \in \Psi$ 
whether the formula $\varphi_0 \wedge \psi^<$ 
is still satisfiable (again, using known methods).  
If $\varphi_0 \wedge \psi^<$ is unsatisfiable, then 
every solution of $\varphi_0$ must satisfy the formula $\psi^=$ obtained from $\psi$ by replacing $\leq$ with $=$. 
We then recursively run the entire algorithm on the formula where we replace the conjunct $\psi$ by $\psi^=$. 
Otherwise, if for every $\psi \in \Psi$, the formula 
$\varphi_0 \wedge \psi^<$ has a solution $s_\psi$, then 
$\varphi_0^{<}$ 
has a solution $s^<$ as well.
This is clear if $\Psi = \emptyset$; 
otherwise, we note that the function $f \colon \mathbb{Q}^k\rightarrow\mathbb{Q}$ given by  $(x_1,\dots,x_k) \mapsto \frac{1}{k} \sum_{i=1}^k x_i$ applied componentwise preserves $+$, $1$, $\leq$,
and strongly preserves $<$ in the sense that 
$f(x_1,\dots,x_k) < f(y_1,\dots,y_k)$ if 
$x_1 \leq y_1$, \dots, $x_k \leq y_k$ and $x_i < y_i$ for at least one $i \in \{1,\dots,k\}$.
This shows that we may take $s^< := \frac{1}{|\Psi|} \sum_{\psi \in \Psi} s_\psi$.

We run the polynomial-time algorithm 
from Theorem~\ref{thm:class2} on $\varphi_=\wedge\varphi_2$
and for each $p\in P\setminus\{2\}$ the polynomial-time algorithm 
from Theorem~\ref{thm:class} on 
$\varphi_=\wedge\varphi_p$.
If one of these algorithms returns NO, then $\varphi$ is unsatisfiable by Proposition~\ref{prop:qvsqp}. 
If all of the algorithms return YES, then 
$\varphi^<$ 
has a solution by Proposition~\ref{prop:qvsqp} and Proposition~\ref{prop:combine},
and therefore also $\varphi$ has a solution.

Finally, 
$\CSP(\mathfrak{Q})$ is in NP
as can be shown 
by repeating the argument from the previous paragraphs for an instance $\varphi$ of $\CSP(\mathfrak{Q})$
and
using Corollary~\ref{cor:NP} instead of the polynomial-time algorithms.
\end{proof}

\section{Conclusions and an open problem}

\noindent
We have presented polynomial-time algorithms for the satisfiability problem of systems of 
linear equalities combined with 
various valuation constraints. For such systems, the satisfiability in ${\mathbb Q}_p$ is equivalent to satisfiability in ${\mathbb Q}$ (Proposition~\ref{prop:qvsqp}). 
We also prove the matching NP-hardness results, answering open questions from~\cite{GHW19} (Theorem~\ref{thm:class} and Theorem~\ref{thm:class2}; also see Figure~\ref{fig:table}). 
Our results can be combined with the polynomial-time tractability result for the satisfiability of (strict and weak) linear inequalities over ${\mathbb Q}$, and we may even solve valuation constraints for different prime numbers simultaneously (Theorem~\ref{thm:combine}). 
Our polynomial-time tractability result for linear inequalities with
valuation constraints of the form $v_2(x) = c$, for constants $c \in {\mathbb Z}$ given in binary, would also follow from a positive answer to the following question, which remains open.

\begin{question}
Is there a polynomial-time algorithm for the satisfiability problem of systems of weak linear inequalities where the coefficients of the inequalities are of the form $2^c$ where $c$ is represented in binary?
\end{question}

Such an algorithm would also imply a polynomial-time algorithm for mean-payoff-games (see~\cite{BodLohoSkomra} for related reductions) which is a problem currently not known to be in P.

\begin{figure}
\begin{center}
\begin{tabular}{|l|p{2.5cm}|p{2.5cm}|p{2.5cm}|p{2.5cm}|}
 \hline 
 & $\mathbb{Q}_p$, $p \geq 3$  
 & $\mathbb{Q}_p$, $p=2$ 
 & $\mathbb{Z}_p$, $p \geq 3$ 
 & $\mathbb{Z}_p$, $p=2$ \\
 \hline 
 $\emptyset$ 
 & \multicolumn{2}{c|}{P: Gauss algorithm}
 & \multicolumn{2}{c|}{P: Hermite normal form} \\ 
 \hline
 $v_p(x)\geq c$ 
 & \multicolumn{4}{c|}{P:  
 \ref{thm:alg-geq} } \\
 \hline
 $v_p(x)=0$ 
 & NP-hard: def.~$\mathbb{Z}_p$ \ref{lem:geq0} 
 & P: reduce to $v_p(x)\geq0$  \ref{lem:eq0} 
 & NP-hard: \ref{prop:eq0} 
 & P: reduce to $\emptyset$ \ref{lem:eq0}  \\ 
 \hline
 $v_p(x)=c$ 
 & NP-hard: solves $v_p(x)=0$ 
 & P: \ref{thm:alg-geq} 
 & NP-hard: solves $v_p(x)=0$ 
 & P: \ref{thm:alg-geq} \\
 \hline
 $v_p(x)\leq 0$ 
 & \multicolumn{2}{c|}{P: special case of $v_p(x)\leq c$} 
 & NP-hard: same as $v_p(x)=0$ 
 & P: same as $v_p(x)=0$\\
 \hline
  $v_p(x)\leq 1$ 
  & \multicolumn{2}{c|}{P: special case of $v_p(x)\leq c$}
  & \multicolumn{2}{c|}{NP-hard: \ref{prop:eq0}}\\
  \hline
 $v_p(x)\leq c$ 
 &\multicolumn{2}{c|}{ P: \ref{prop:alg-leq}}
 &\multicolumn{2}{c|}{ NP-hard: solves $v_p(x)\leq 1$} \\
 \hline
 $v_p(x)\neq 0$ 
 &\multicolumn{2}{c|}{P: \ref{prop:alg-leq} or reduce to $v_p(x)\leq 0$ via \ref{lem:neq0}}
 &\multicolumn{2}{c|}{P: reduces to $\emptyset$ via \ref{lem:neq0}}\\
 \hline
 $v_p(x)\neq c$ 
 &\multicolumn{2}{c|}{P: \ref{prop:alg-leq}}
 &\multicolumn{2}{c|}{NP-hard: def.~$v_p(x)\leq 1$ via \ref{lem:leq0}}\\
 \hline
\end{tabular}
\end{center}
\caption{An overview of polynomial-time tractability and NP-hardness for systems of linear equations with valuation constraints.}
\label{fig:table}
\end{figure}

\bibliographystyle{alpha}
\bibliography{global}

\end{document}